\documentclass[letterpaper, 10pt, conference]{ieeeconf}

\IEEEoverridecommandlockouts
\usepackage{cite}
\usepackage[utf8]{inputenc} 
\usepackage{amsmath} 
\usepackage{amssymb}  
\usepackage{amsfonts}  
\usepackage{graphics}
\usepackage{pdflscape}
\usepackage{rotating}
\usepackage{arcs}
\usepackage{epstopdf}
\usepackage{subcaption}

\newtheorem{assumption}{Assumption}
\newtheorem{lemma}{Lemma}
\newtheorem{remark}{Remark}

\newtheorem{corollary}{Corollary}
\newtheorem{proposition}{Proposition}

\def\begequarr{\begin{eqnarray}}
\def\endequarr{\end{eqnarray}}
\def\begequarrs{\begin{eqnarray*}}
\def\endequarrs{\end{eqnarray*}}
\def\begarr{\begin{array}}
\def\endarr{\end{array}}
\def\begequ{\begin{equation}}
\def\endequ{\end{equation}}
\def\lab{\label}
\def\begdes{\begin{description}}
\def\enddes{\end{description}}
\def\begenu{\begin{enumerate}}
\def\begite{\begin{itemize}}
\def\endite{\end{itemize}}
\def\endenu{\end{enumerate}}

\def\lef[{\left\{\begin{array}}
\def\rig]{\end{array}\right\}}

\def\begcen{\begin{center}}
\def\endcen{\end{center}}
\def\begrem{\begin{remark}\rm}
\def\endrem{\end{remark}}
\def\begsubequ{\begin{subequations}}
\def\endsubequ{\end{subequations}}

\def\calj{{\cal J}}

\def\call{{\cal L}}

\def\calj{{\cal J}}

\def\liminf{\lim_{t \to \infty}}

\def\L2e{{\cal L}_{2e}}
\def\bul{\noindent $\bullet\;\;$}
\def\rea{\mathbb{R}}

\def\adj{\mbox{adj}}
\def\col{\mbox{col}}

\def\hal{{1 \over 2}}
\def\et{\varepsilon_t}

\def\begmat#1{\begin{bmatrix}#1\end{bmatrix}}
\def\begali#1{\begin{align}{#1}\end{align}}
\def\begalis#1{\begin{align*}{#1}\end{align*}}

\def\CEP{{\it Control Engineering Practice}}

\def\TAC{{\it IEEE Trans. Automatic Control}}

\def\SCL{{\it Systems \& Control Letters}}
\def\AUT{{\it Automatica}}

\usepackage{color}

\usepackage[prependcaption,colorinlistoftodos]{todonotes}

\graphicspath{{fig/}
				{fig/test1/} 
				{fig/test2/} 
				{fig/test3/}}
\title{\LARGE \bf
A Flux and Speed Observer for Induction Motors with Unknown Rotor Resistance and Load Torque and no Persistent Excitation Requirement
}
\author{Anton Pyrkin$^{1, 2}$, Alexey Bobtsov$^{1, 2}$, Alexey Vedyakov$^{1}$, Romeo Ortega$^{1,3}$,\\ Anastasiia Vediakova$^{4}$, Madina Sinetova$^{1}$ %
\thanks{*This article is supported in ITMO University under the Ministry of Science and Higher Education of Russian Federation (goszadanie 2019-0898) and by Government of Russian Federation (grant 08-08).}
\thanks{$^{1}$ Faculty of Control Systems and Robotics, ITMO University, Kronverksky av., 49, 197101, Saint Petersburg, Russia, 
{\tt\small vedyakov@itmo.ru}}
\thanks{$^{2}$ Center for Technologies in Robotics and Mechatronics Components, Innopolis University, Innopolis, Russia}
\thanks{$^{3}$ Departamento Acad\'{e}mico de Sistemas Digitales, ITAM, Ciudad de M\'exico, M\'{e}xico}
\thanks{$^{4}$ Department of Computer Applications and Systems, St.Petersburg State University, 7/9 Universitetskaya nab., St. Petersburg, 199034,~Russia}
\thanks{Conflict of interest --- none declared}
}
\begin{document}

\maketitle
%
\begin{abstract}
In this paper we address the problems of flux and speed observer design for voltage-fed induction motors with unknown rotor resistance and load torque. The only measured signals are stator current and control voltage. Invoking the recently reported Dynamic Regressor Extension and Mixing-Based Adaptive Observer (DREMBAO) we provide the first global solution to this problem. The proposed DREMBAO achieves {\em asymptotic} convergence under an excitation condition that is strictly weaker than persistent excitation. If the latter condition is assumed the convergence is {\em exponential}. 
\end{abstract}

\section{INTRODUCTION}
\lab{sec1}
Because of its great practical and theoretical importance control of  induction motors (IM) has attracted much attention from researchers and engineers  for over 50 years now. More than 5,000 journal papers have been published on IM control, being to date still a very active research area. The industrial interest in IM control is documented by over 80,000 patents on this subject. In spite of the intense research efforts in the field of IM control there are several important problems that remain open, {\em cf.}, \cite{MARTOMVERbook,ORTetalbook}.

We address in the paper the problems of estimation of the rotor resistance and the load torque, as well as the design of flux and speed observers in the absence of the knowledge of these parameters. Providing an answer to these questions is relevant for the solution of the so-called sensorless control problem as well as in fault detection and motor calibration tasks. For a review of the literature the reader is referred to the excellent, comprehensive research monograph \cite{MARTOMVERbook}, see also \cite{ASTKARORTbook,NAMbook,ORTetalbook}.

In this paper we give solutions to the following.\\ 

\noindent {\bf Adaptive Observer Problem} Given the $5$th-order dynamics of the {\em voltage-fed} IM with
\begite
\item {\em measurable} stator current and voltage;
\item {\em known} stator inductance and resistance and leakage coefficient;
\item {\em unknown} rotor resistance and mechanical load torque.
\endite
Design an observer for the rotor flux and the speed which ensures global {\em asymptotic} convergence of the unknown parameters and unmeasurable states under excitation conditions that are {\em strictly weaker} than the classical persistent excitation (PE) requirement \cite[Section 2.5]{SASBOD}.\\  

\noindent {\bf State of the Art} Many authors have studied these problems, under different assumptions, and adopting various approaches including: high-gain based techniques, like sliding modes; designs based on linear approximations, like Kalman filtering and Model Reference Adaptive Systems; and schemes based on Neural Network or Fuzzy Control. 

We concentrate in this paper on results for which a rigorous {\em mathematical} proof, under reasonable, verifiable assumptions is provided. In this sense, to the best of our knowledge, the aforementioned questions are open, and we provide in this paper the first solutions to them. The requirement of ``reasonable, verifiable assumptions" leads us to rule out schemes based on open-loop integration of IM currents and/or voltages, {\em cf.}, \cite{BAZetalwc17,VERetalcep17}, which is not practically feasible. 

Several solutions for particular cases of the problem are known and some of them are reviewed below.\\ 

\noindent \bul The following results assume the motor {\em speed is measurable}.\footnote{We refer the interested reader to the quoted monograph references to find out the journal where these results were first reported.} 

\noindent - In \cite[Subsection 3.2]{MARTOMVERbook} a rotor flux observer that estimates the rotor resistance is proposed. The observer has a redundant dynamics and convergence is guaranteed under a PE condition imposed on some of the estimated signals \cite[Equation (3.59)]{MARTOMVERbook}.

\noindent - In \cite[Subsection 10.3]{ASTKARORTbook}  a rotor flux observer that estimates the load torque, assuming known the rotor resistance is proposed. The observer is proposed as part of a globally convergent speed tracking controller that does not require any excitation assumption.

\noindent -  A load torque estimator, assuming the rotor flux can be recovered exponentially fast, is proposed in  \cite[Subsection 3.2]{MARTOMVERbook}.

\noindent -  An adaptive observer-based speed control with uncertain load torque that estimates the rotor flux assuming all machine parameters known is proposed in  \cite[Subsection 4.3]{MARTOMVERbook}. This result is extended to the case of unknown rotor resistance with known lower bound. 

\noindent \bul For the case when rotor speed is {\em not measurable} we are only aware of the result of \cite[Subsection 5.4]{MARTOMVERbook} where a local result assuming known the rotor resistance and some PE conditions is reported.\\

\noindent {\bf About the Paper} To provide a solution to the adaptive observer problem stated above we rely on three essential components.
\begenu
\item[{\bf C1}] The key observation made in \cite[Subsection 1.2]{MARTOMVERbook} that the derivatives of the flux and the current are related by a simple relation. This observation was already used in \cite{BAZetalwc17} to address the adaptive observer problem but, a practically inadmissible, open-loop integration of the IM currents and voltages was proposed. 
\item[{\bf C2}] Proceeding from a reparameterization of the aforementioned equation---using the norm of the rotor flux---we derive a {\em linear regression equation} (LRE), which involves unknown parameters, unmeasurable states and their {\em product}. This step was first reported in \cite{PYRetalcdc19}.
\item[{\bf C3}]  The use of DREMBAO \cite{PYRetalscl19}, which is an advanced technique for the design of adaptive observers able to handle the presence of {\em products} between unknown parameters and unobservable states. Towards this end, DREMBAO uses the {\em dynamic regressor extension and mixing} (DREM) parameter estimation procedure \cite{ARAetaltac17,ORTetaltac20} to generate scalar regressions, using which we obtain the estimated parameters and the observed states. 
\endenu 

The remainder of the paper is organized as follows. In Section~\ref{sec2} we present the model of the IM and the problem formulation. The novel parameterization of the IM model and its associated vector LRE are described in Section~\ref{sec3}. In Sections~\ref{sec4} and \ref{sec5} we present the new flux observer and the rotor resistance estimator, respectively. Using these flux and rotor resistance estimates in Section~\ref{sec6} we propose the new speed observer and load torque estimator. In Section~\ref{sec7} we present simulation results, which demonstrate the effectiveness of the proposed approach. Our work is wrapped-up with some conclusions and future research in Section~\ref{sec8}.\\

\noindent {\bf Notation.} $I_n$ is the $n \times n$ identity matrix. We use $\rea_+:=(0,\infty)$. For $x \in \rea^n$, we denote the Euclidean norm $|x|^2:=x^\top x$. All mappings are assumed smooth.  For an LTI filter $G(p) \in \rea(p)$ its action on a signal $w(t)$ is denoted $G(p)[w]$.  
\section{PROBLEM FORMULATION}
\lab{sec2}
%
Consider the electrical dynamics of the fixed-frame model of the voltage-fed induction motor~\cite[Equation (4.42)]{NAMbook}
\begsubequ
\lab{imdyn}
\begali{
\dot\lambda &= -\left({R_r\over L_r}I_2-n_p \calj\omega\right)\lambda+R_r\beta i,
\lab{im1_flux}
\\
L_s \sigma{di\over dt}&=  - (R_s+R_r\beta^2) i +\beta\left({R_r \over L_r} I_2 - n_p \calj\omega \right)\lambda+ v,
\lab{im1_current}
}
\endsubequ
where $\lambda,i,v \in \rea^2$ are the rotor flux, the stator current and the control voltage, respectively, $\omega \in {\mathbb D}$ is the rotor speed, $L_s,R_r,L_r,R_s,n_p,\sigma$ are positive constants representing the stator inductance, rotor resistance, rotor inductance, stator resistance, number of pole pairs and leakage parameter, respectively. To simplify the notation, we defined $\beta:={M\over L_r}$, where $M$ is the mutual inductance, and  
$$
\calj:=\begmat{0 & -1 \\ 1 & 0}.
$$ 

The mechanical dynamics, on the other hand, is described by
 \begin{align}
J\dot\omega &=-n_p \beta \lambda^\top  \calj i-T_L,
\lab{im1_mecdyn}
\end{align}
where $J$ is the rotor inertia and $T_L$ is the load torque, which is assumed {\em constant}.

The goal is to design an observer for the rotor flux $\lambda$ and the speed $\omega$, assuming only the current $i$, and the voltage $v$ are measured, that the electrical parameters $L_s$, $R_s$ and $\sigma$ are known, but $R_r$ and $T_L$ are {\em unknown}. As discused in the Introduction the importance of this problem can hardly be overestimated. 

Consistent with the usual observer design scenario \cite{PBERbook}, we assume that the external signals $v$ and $T_L$ are such that the system~\eqref{im1_flux}--\eqref{im1_mecdyn} is forward complete and all the signals are bounded. Furthermore, we also assume that $v$ and $i$ are absolutely integrable. This assumption is consistent with the motor operation since, in steady-state, these signals are periodic of zero-mean.
%
\section{A LINEAR REGRESSION EQUATION}
\label{sec3}
%
As explained in  \cite{PYRetalscl19} the key step for the application of DREMBAO is to derive a {\em family of parameterized LREs} to which we can apply the DREM procedure to isolate its scalar components, namely \cite[Equation (8)]{PYRetalscl19} and  \cite[Equation (11)]{PYRetalscl19}. 

In the present case, this LRE is identified in the lemma below, whose proof---being quite technical---is given in the Appendix. 

\begin{lemma}
\lab{lem1}
Consider the IM electrical dynamics \eqref{imdyn} with measured signals
$$
y:=\col(i,v).
$$
There exists {\em measurable} signals\footnote{That is, signals that can be {\em computed}, via stable filtering and algebraic operations, from the measured signals $y$---without open-loop integration nor differentiation.}
\begalis{
z_e &:\rea_+ \times \rea_+ \to \rea \\
\phi_e &:\rea_+ \times \rea_+ \to \rea^6,
}
verifying
\begequ
\lab{lre}
z_e(t,\alpha_\ell )=\phi_e^\top (t,\alpha_\ell )\begmat{R_r\\ \lambda(t)\\ R_r\lambda(t) \\ R_r|\lambda(t)|^2} +\et,
\endequ
where $\alpha_\ell >0$ is a designer-chosen parameter and $\et$ is a generic, exponentially decaying signal stemming from some LTI filters initial conditions.
\end{lemma}

\begrem
\lab{rem1}
We underscore the presence of {\em products} between the unknown parameter $R_r$ and the state to be reconstructed $\lambda$ in \eqref{lre}, which makes the adaptive observation problem unsolvable with standard techniques \cite{ASTKARORTbook,PBERbook,MARTOMVERbook}. 
\endrem
%
\section{GENERATION OF SCALAR REGRESSIONS}
\label{sec4}
%
In this section we apply the DREM methodology to generate, from the vector LRE \eqref{lre}, six {\em scalar} LRE. In this way we ``isolate" two scalar LRE for the flux $\lambda$ and the unknown parameter $R_r$, that we can easily identify. The result is contained in the lemma below, whose proof is given in the Appendix. 

\begin{lemma}
\lab{lem2}
Consider the family of LREs \eqref{lre}. Fix six different, positive constants $\alpha_\ell,\;\ell=1,\dots,6$. There exists {\em measurable} signals
\begalis{
\zeta_e &:\rea_+ \to \rea^6 \\
\Delta_e &:\rea_+ \to \rea,
}
verifying
\begequ
\lab{scalre}
\zeta_e(t)=\Delta_e(t)\begmat{R_r\\ \lambda(t)\\ R_r\lambda(t) \\ R_r|\lambda(t)|^2} +\et.
\endequ
\end{lemma}

\vspace{0.5cm}

\begrem
\lab{rem2}
We underline the fact that $\Delta_e(t)$---defined in \eqref{deltae} in the proof of Lemma \ref{lem2}---is a {\em scalar} signal. Consequently, from the first three elements of \eqref{scalre} we can define three scalar LREs
\begali{
\lab{lre1}
\zeta_{e1}(t)&=\Delta_e(t)R_r +\et\\
\lab{lre23}
\zeta_{e23}(t):=\begmat{\zeta_{e2}(t)\\ \zeta_{e3}(t)}&=\Delta_e(t)\lambda(t) +\et,
}
from which we can estimate $R_r$ and reconstruct $\lambda$ {\em independently}. A task that is carried out in the next two sections.
\endrem 
%
\section{FLUX OBSERVER}
\label{sec5}
%
Our first main result, that is, a globally convergent observer for the flux, is given in the following proposition. To establish the result we need an excitation assumption articulated below.

\begin{assumption}
\lab{ass1}
Assume that the scalar signal $\Delta_e$,  defined in  \eqref{deltae} in the proof of Lemma \ref{lem2}, verifies
$$
\Delta_e(t) \notin \call_2 \quad \Longleftrightarrow\quad\liminf \int_0^t\Delta_e^2(s)ds=\infty.
$$
\end{assumption}

\vspace{0.5cm}

\begin{proposition}
	\label{pro1}
	Consider the model of the IM \eqref{imdyn} and the LRE~\eqref{lre23} with $\Delta_e$ verifying Assumption \ref{ass1}. The flux observer
	\begin{align}\
	\label{eq:flux_observer_chi}
		\dot{\chi}_e &= \frac{1}{\beta}v - \frac{R_s}{\beta} i + \gamma_{\lambda} \Delta_e \Big[\zeta_{e23} +\Big(\frac{\sigma L_s}{\beta}i - \chi_e\Big) \Delta_e\Big]\\
		\hat{\lambda} &= -\frac{\sigma L_s}{\beta}i + \chi_e,
		\label{eq:flux_observer}
	\end{align}
where $\gamma_{\lambda}>0$ is a tuning gain, ensures
$$
\liminf |\lambda(t) - \hat \lambda(t)|=0.
$$

\end{proposition}

\begin{proof}
	Define the observation error
	\begin{align}
		\tilde{\lambda} := \lambda - \hat{\lambda},
	\end{align}
	using the flux dynamics equation~\eqref{im2}, the observer equations \eqref{eq:flux_observer_chi} and \eqref{eq:flux_observer}, and the LRE\footnote{For ease of presentation we neglect the exponentially decaying term in  \eqref{lre23}. See Remark 5.} \eqref{lre23} we get the error dynamics
	\begin{align*}
		\dot{\tilde{\lambda}}
			& = \dot{\lambda} - \dot{\hat{\lambda}} 
			\nonumber \\
			& = \dot{\lambda} + \frac{\sigma L_s}{\beta}{di\over dt} - \dot{\chi}_e
			\nonumber \\
			& = - \gamma_{\lambda} \Delta_e(\zeta_{e23} - \hat{\lambda} \Delta_e)
			\nonumber \\
			& = - \gamma_{\lambda}\Delta_e^2\tilde{\lambda}.
	\end{align*}
The solution of this scalar differential equation is
\begequ
\lab{soloode}
\tilde{\lambda}(t)= e^{-\gamma_\lambda \int_0^t\Delta_e^2(s)ds}\tilde{\lambda}(0)
\endequ
from which we conclude the proof.
\end{proof}

\begrem
\lab{rem3}
It has been shown in \cite{ORTetaltac20} that the condition $\Delta_e(t) \notin \call_2$ is strictly weaker than PE of the regressor $\phi_e(t)$ of \eqref{lre}.
\endrem

\begrem
\lab{rem4}
In \cite{YIORT} it is established that---generically, {\em i.e.}, for almost all choices of the constants $\alpha_\ell$---if $\phi_e(t)$ is PE then $\Delta_e(t)$ is also PE. It is easy to see that, in that case, there exists positive constants $C_{\lambda}$ and $\rho_{\lambda}$ such that
	$$
	|\tilde \lambda(t)| \leq C_{\lambda} e^{-\rho_{\lambda}t}.
	$$
\endrem

\begrem
\lab{rem5}
Notice that the additive exponentially decaying term $\et$ in \eqref{lre23}, that we neglected in the proof above, appears in the flux observer error equation in the form
$$
\dot{\tilde{\lambda}} = - \gamma_{\lambda}\Delta_e^2\tilde{\lambda}+\gamma_{\lambda}\Delta_e \et.
$$
As shown in \cite[Lemma 1]{ARAetalmicnon15} the presence of the term $\et$ does not affect the result of Proposition \ref{pro1}. Therefore, in the sequel, we neglect the presence of such terms.
\endrem
\section{ROTOR RESISTANCE ESTIMATION}
\label{sec6}
%
In this section the rotor resistance is estimated using the LRE~\eqref{lre1}.
\begin{proposition}
	\label{pro2}
	Consider the LRE~\eqref{lre1} and the parameter update law
	\begin{align}
	\label{eq:rr_est}
	\dot{\hat{R}}_r = \gamma_{r}  \Delta_e \left( \zeta_{e1} - \hat{R}_{r}  \Delta_e \right),
	\end{align}
	where $\gamma_{r} > 0$. If $ \Delta_e$ verifies Assumption \ref{ass1} then
$$
\liminf |\hat R_r(t)-R_r|=0.
$$
Moreover, if $\Delta_e$ is PE then there exists positive constants $C_{R}$ and $\rho_{R}$ such that
	\begin{align}
	|\hat{R}_r(t) - R_{r}| \leq C_{R} e^{-\rho_{R}t}.
	\end{align}
\end{proposition}
\begin{proof}
	Defining the observation error
	\begin{align}
	\tilde{R}_r = \hat{R}_r - R_r,
	\end{align}
	and substituting~\eqref{eq:rr_est} we obtain
	\begin{align}
	\dot{\tilde{R}}_r
	= - \dot{\hat{R}}_r 
	& = - \gamma_{r} \Delta_e^2(t) \tilde{R}_{r}.
	\label{eq:r_error_dyn}
	\end{align}
The proof is completed using the same arguments of the proof of Proposition \ref{pro1}.
\end{proof}
%
\section{ROTOR SPEED OBSERVER AND LOAD TORQUE ESTIMATION}
\label{sec7}
%
In this section we design an observer for the speed $\omega$ and an estimator for the load torque $T_L$. In the light of Propositions \ref{pro1} ans \ref{pro2}, we apply {\em certainty equivalence} and---assuming $\Delta_e$ verifies Assumption \ref{ass1}---consider that the resistance and the flux are obtained applying these propositions. The constructions are done with a procedure similar to the one used above. Namely, doing first some filtering to obtain a vector LRE in the unknowns $\omega$ and $T_L$. Then, using DREM to derive independent, scalar LREs for $\omega$ and $T_L$.  
\subsection{Derivation of two scalar LREs}
\label{subsec71}
%
\begin{lemma}
\lab{lem3}
Consider the IM electrical~\eqref{imdyn} and mechanical \eqref{im1_mecdyn} dynamics with known flux $\lambda$ and rotor resistance $R_r$. There exists {\em measurable} signals
\begalis{
\zeta_m &:\rea_+  \to \rea^2 \\
\Delta_m &:\rea_+ \to \rea,
}
verifying
\begequ
\lab{lreomet}
\zeta_m(t)=\Delta_m(t)\begmat{T_L\\ \omega(t)} +\et.
\endequ
\end{lemma}
\begin{proof}
The first step in the proof is to derive a vector LRE for $\col(T_L,\omega)$. For, consider the equations of the IM~model~\eqref{im1_flux} and \eqref{im1_mecdyn}, that we rewrite as
\begali{
\lab{dotlam}
	\dot\lambda +\eta_1
	& = \eta_2\,\omega,
	\\
\label{jdotome}
	J\dot\omega &=\beta \eta_2^\top i-T_L,
}
where we defined the two-dimensional, measurable signals
\begin{align*}
\eta_1 & := \frac{R_r}{L_r}\lambda - R_r\beta i,
\\
\eta_2 & := n_p \calj \lambda,
\end{align*}
Applying the filter $\frac{a}{p+a}$, with $a > 0$, and the Swapping Lemma to \eqref{dotlam} we obtain
\begalis{
&\frac{ap}{p+a}[\lambda]+\frac{a}{p+a}[\eta_1]=\frac{a}{p+a}[\eta_2 \omega]\\
&=\frac{a}{p+a}[\eta_2]\omega-\frac{1}{p+a}\Big[ \dot \omega\frac{a}{p+a}[\eta_2]\Big].
}
Using \eqref{jdotome} to replace $\dot \omega$ in the right hand side term above yields the matrix LRE
\begin{align}
\label{eq:reg2}
z_m(t) =  \Phi_{m}(t)\begmat{T_{L}\\  \omega(t)}+\et,
\end{align}
where we defined the measurable signals
\begin{align*}
z_m & :=\frac{a p}{p+a} [\lambda]+ \frac{a}{p+a}[\eta_1]+ \frac{\beta }{J}\frac{1}{p+a} \left[\eta_2^{\top} i \frac{a}{p + a} [\eta_2] \right]
\\
\Phi_{m} & := \begmat{ \frac{1}{J}\frac{a}{(p +a)^2}[\eta_2] & | & \frac{a}{p+a}[\eta_2]}.
\end{align*}

Now, we apply DREM to \eqref{eq:reg2} {and} obtain \eqref{lreomet} with the definitions
\begin{align}
\nonumber
\zeta_m & := \adj\{\Phi_m\}z_m\\
\lab{delm}
\Delta_m & := \det\{\Phi_m\}.
\end{align}
completing the proof.
\end{proof}

\begrem
\lab{rem6}
Notice that the matrix $\Phi_m$ may be written as
\begequ
\lab{phim}
\Phi_m=\frac{a}{p+a}\begmat{\frac{1}{J(p +a)}[\calj \lambda] & | & \calj \lambda }
\endequ
underscoring the critical role of the flux vector in the excitation requirement. See Corollary \ref{cor1} below.
\endrem
\subsection{Estimation of $T_L$ and observation of $\omega$}
\label{subsec72}
%
To design these estimators we need an additional excitation assumption articulated below.

\begin{assumption}
\lab{ass2}
The scalar signal $\Delta_m$ defined in \eqref{delm} verifies
$$
\Delta_m(t) \notin \call_2 \quad \Longleftrightarrow\quad\liminf \int_0^t\Delta_m^2(s)ds=\infty.
$$
\end{assumption}
\begin{proposition}
	\label{pro3}
	Consider the IM electrical~\eqref{imdyn} and mechanical \eqref{im1_mecdyn} dynamics with known flux $\lambda$ and rotor resistance $R_r$, with $\Delta_m$ defined in \eqref{delm} verifying Assumption \ref{ass2}. The parameter estimator
\begin{align}
\dot{\hat{T}}_L = \gamma_{T} {\Delta}_m \left[ \zeta_{m1} - \hat{T}_{T} {\Delta}_m \right],
\end{align}
together with the speed observer
\begin{align*}
\dot{\hat \omega} &= -\frac{1}{J}\hat{T}_L-\frac{n_p \beta }{J}\lambda^{\top} \calj i+\gamma_{\omega} \hat{\Delta}_m \left( \zeta_{m2} -  \hat \omega {\Delta}_m \right)
\end{align*}
where $\gamma_{T} > 0$ and $\gamma_{\omega} > 0$ are tuning parameters, ensures
\begalis{
\liminf |\hat T_L(t)-T_L| & =0\\
\liminf |\hat \omega(t)-\omega(t)| & =0
}
Moreover, if $\Phi_m(t)$ is PE the convergences are exponential. 
\end{proposition}
\begin{proof}
The proof is established invoking \eqref{lreomet} and verifying that the error equations take the forms
\begalis{
\dot{\tilde{T}}_L
	& = - \gamma_{T} \Delta_m^2 \tilde{T}_{L}.\\
	\dot{\tilde{\omega}}
	& = - \gamma_{\omega} \Delta_m^2 \tilde \omega-\frac{1}{J}\tilde{T}_L,
}
and using arguments similar to the once used in the proof of Proposition \ref{pro1} and standard cascaded systems analysis.
\end{proof}

The following corollary proves that, in steady state, $\Delta_m$ is PE, ensuring exponential convergence of the estimates of $T_L$ and $\omega$.  

\begin{corollary}
\lab{cor1}
Consider the IM electrical~\eqref{imdyn} and mechanical \eqref{im1_mecdyn} dynamics operating in {\em steady-state}. The adaptive observer of Proposition \ref{pro3} is {\em exponentially} convergent.
\end{corollary}

\begin{proof}
The steady-state time-varying operation of the IM with sinusoidal voltages is given by \cite[equation (1.47)]{MARTOMVERbook}
\begalis{
\lambda^\star &=|\lambda^\star |\begmat{\cos(\rho^\star )\\ \sin(\rho^\star )}\\
\dot \rho^\star &=\omega^\star +{R_r T_L \over |\lambda^\star |^2}\\
\rho^\star(0) &=\arctan\Big\{{\lambda_2(0) \over \lambda_1(0)}\Big\},
}
where $(\cdot)^\star $ denotes their reference value. Replacing these values in the matrix $\Phi_m$ given in \eqref{phim}, computing the steady-state values of the filter outputs and using some simple trigonometric identities yields
$$
\Delta_m^\star =-{|\lambda^\star | \over J}\sin(\psi^\star ),
$$
where $\psi^\star $ is the phase shift of the filter ${1 \over j\varpi+a}$ at the frequency 
$$
\varpi=\omega^\star +{R_r T_L \over |\lambda^\star |^2}.
$$
This completes the proof. 
\end{proof}

\begrem
\lab{rem7}
Clearly, for the adaptive implementation of the algorithms of Proposition \ref{pro3} we replace $\lambda$ and $R_r$ by their estimates generated as indicated in Propositions \ref{pro1} and \ref{pro2}, respectively. Due to the complicated algebraic operations involved in the derivation of the LRE \eqref{lre} the mathematical analysis of this implementation of the adaptive observer---without the certainty equivalent assumption---is a daunting task.

\endrem
%
\section{SIMULATION RESULTS}
\label{sec8}
%
The proposed observers and estimators have been tested via numerical simulations in the open loop. The IM is driven by the standard full-state measurement field-oriented control~\cite{leonhard2001control} independently on the flux estimates
\begin{equation}
	v_{dq}=K_p(i_{dq}^\text{ref}-i_{dq})+K_i\int_0^t(i_{dq}^\text{ref}-i_{dq})d\tau,
\end{equation}
where $(\cdot)_{dq}=e^{-J  \delta}(\cdot)_{ab}$,
\begin{align}
	\delta:=\arctan\left(\frac{{\lambda}_b}{{\lambda}_a}\right),
\end{align}
and the current references are generated as
\begin{align}
	i_{d}^\text{ref}&=\frac{1}{M}|{\lambda}_{ab}|+\frac{L_r}{\hat{R}_rM}\left(K_{\lambda p}e_{\lambda}+K_{\lambda i}\int_0^t e_{\lambda}d\tau\right),
	\\
	i_{q}^\text{ref} &=\frac{J_mL_r}{M|{\lambda}_{ab}|}\left(K_{\omega p}e_{\omega}+K_{\omega i}\int_0^t e_{\omega}d\tau\right)
\end{align}
with the error signals 
\begin{align}
\lab{lamref}
	e_{\lambda} &:= |\lambda_{ab}|^\text{ref}-|\lambda_{ab}|,
	\\
	\lab{omref}
	e_{\omega} & = \omega^\text{ref}-\omega,
\end{align}
and the six controller gains $K_{(\cdot)} >0$ are given below.

\begin{figure*}[htp]
	\centering
	\subcaptionbox{\label{fig:lambda_err} Flux estimation error}{\includegraphics[width=0.49\textwidth]{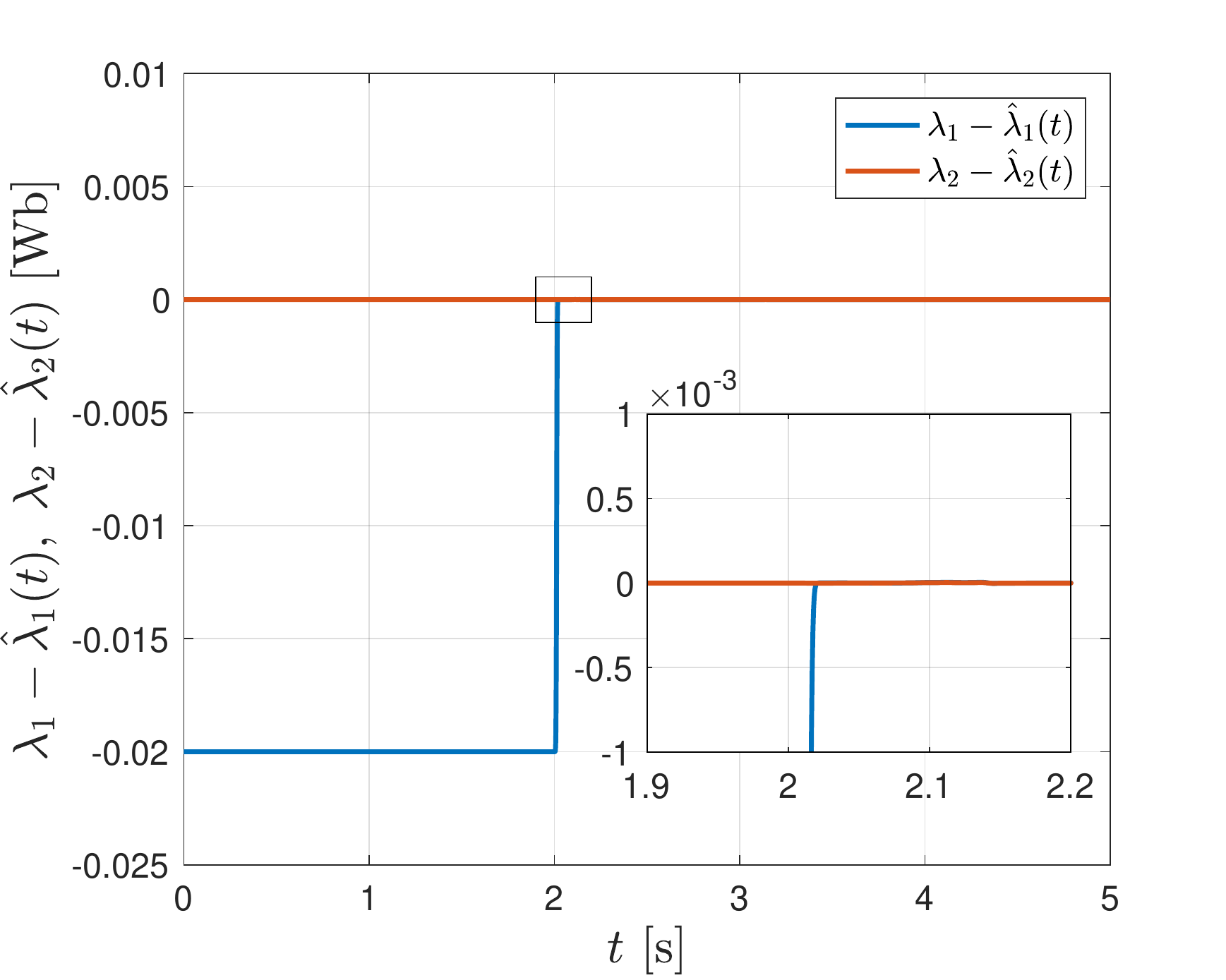}}
	\
	\subcaptionbox{\label{fig:rr} $R_r$, estimate, and error}{\includegraphics[width=0.49\textwidth]{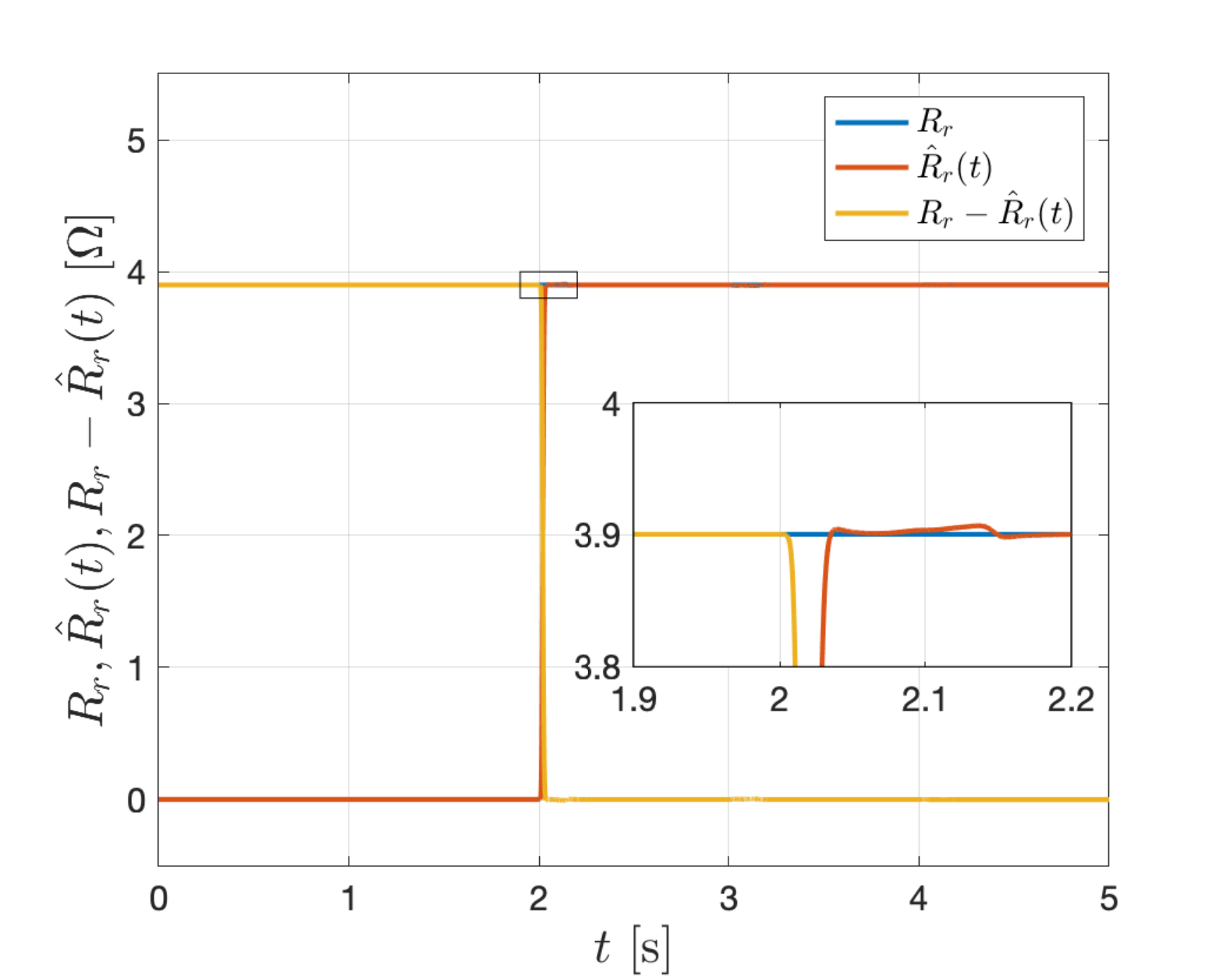}}
	\vspace{-2mm}
	\caption{\label{fig:observer_lambda_R} Transitions for flux observer and $R_r$ estimator }
\end{figure*}

The IM parameters are same as in~\cite{BAZetalwc17}, namely, $L_s = 140$~mH, $L_r = 140$~mH, $M = 117$~mH, $R_s = 1.7\,\mathrm{\Omega}$, $R_r = 3.9\,\mathrm{\Omega}$, $J= 0.00011$ $\mathrm{kg\,m^2}$. The amplitude of the rotor flux reference in \eqref{lamref} was chosen $ |\lambda_{ab}|^\text{ref} =0.0455$~Wb. The speed reference in \eqref{omref} was chosen $ \omega^\text{ref}=40$~rad/s (see Fig.~\ref{fig:omega}), with the following initial conditions $\theta(0) = -3$~rad, $\omega(0) = 0$~rpm, $\lambda(0) = (0.02,0)$~Wb. The controller tuning gains were selected as $K_p =100$, $K_i =100$, $K_{\lambda p} =10$, $K_{\lambda i} =100$, $K_{\omega p} = 10$ and $K_{\omega i} = 10$.

The observers parameters were chosen as: $\gamma_d = 1000$, $\alpha_1 = 10$, $\alpha_2 = 20$, $\alpha_3 = 30$, $\alpha_4 = 40$, $\alpha_5 = 50$, $\alpha_6 = 100$, $\gamma_{\lambda}=0.001$, $\gamma_{r}=0.0001$, $\gamma_{\omega}=10^6$, $\gamma_{T_{L}}=10^6$. The observer starts working after two seconds, with the observer input signals set to zero before that time.

Fig.~\ref{fig:lambda_err} shows the behavior of the flux observer, while the one of the rotor resistance estimator is depicted in Fig.~\ref{fig:rr}. As shown in the figures the performance of the observer and the estimator is remarkable. In Fig.~\ref{fig:omega} the actual rotor speed, its estimate and estimation error are depicted, while the ones of the external load are shown in Fig.~\ref{fig:taul}. Similarly to the previous remark, the quality of the transients is excellent. 

\begin{figure*}[htp]
	\centering
	\subcaptionbox{\label{fig:omega} Transitions for speed estimate}{\includegraphics[width=0.49\textwidth]{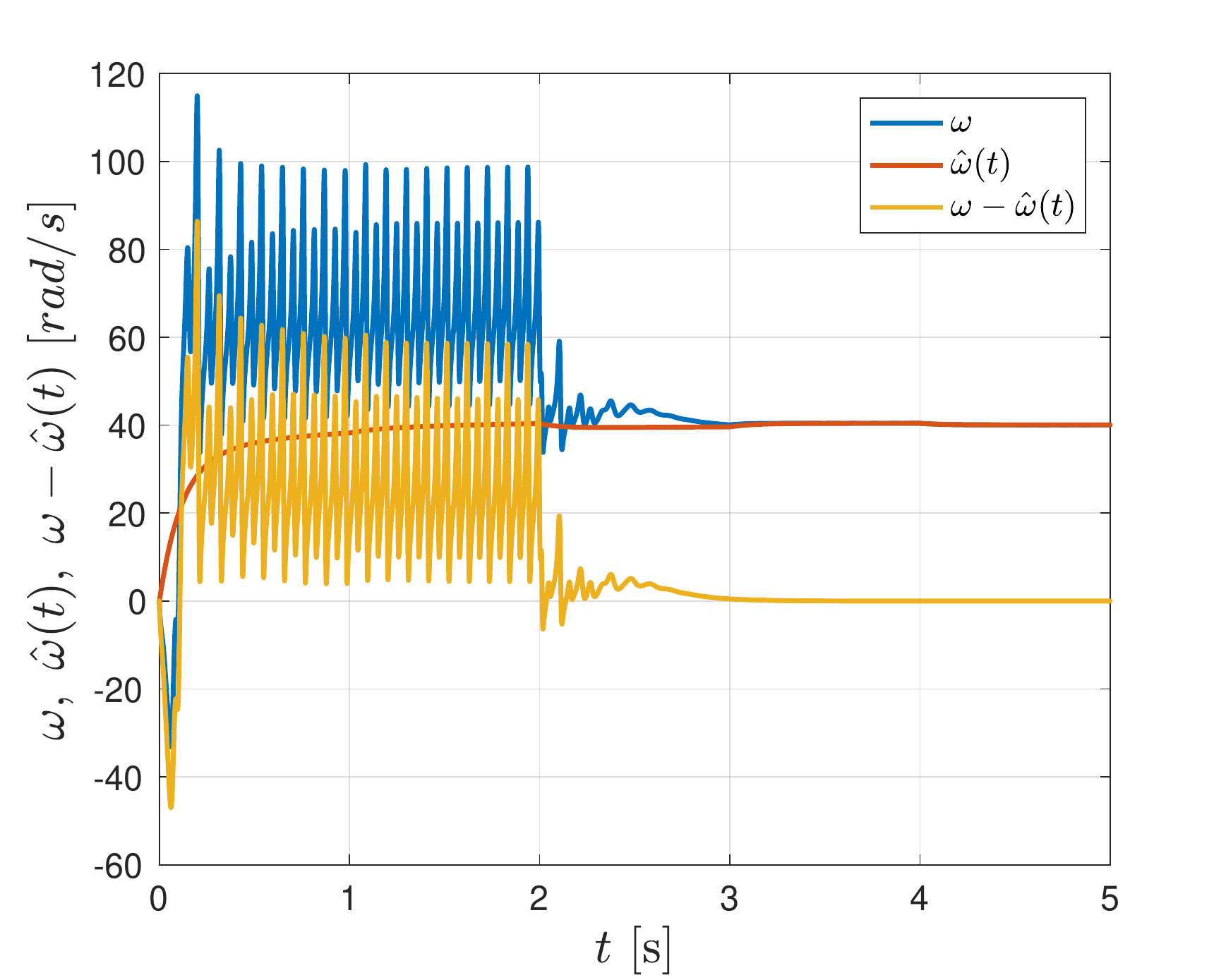}}
	\
	\subcaptionbox{\label{fig:taul} Transitions for external load torque estimate}{\includegraphics[width=0.49\textwidth]{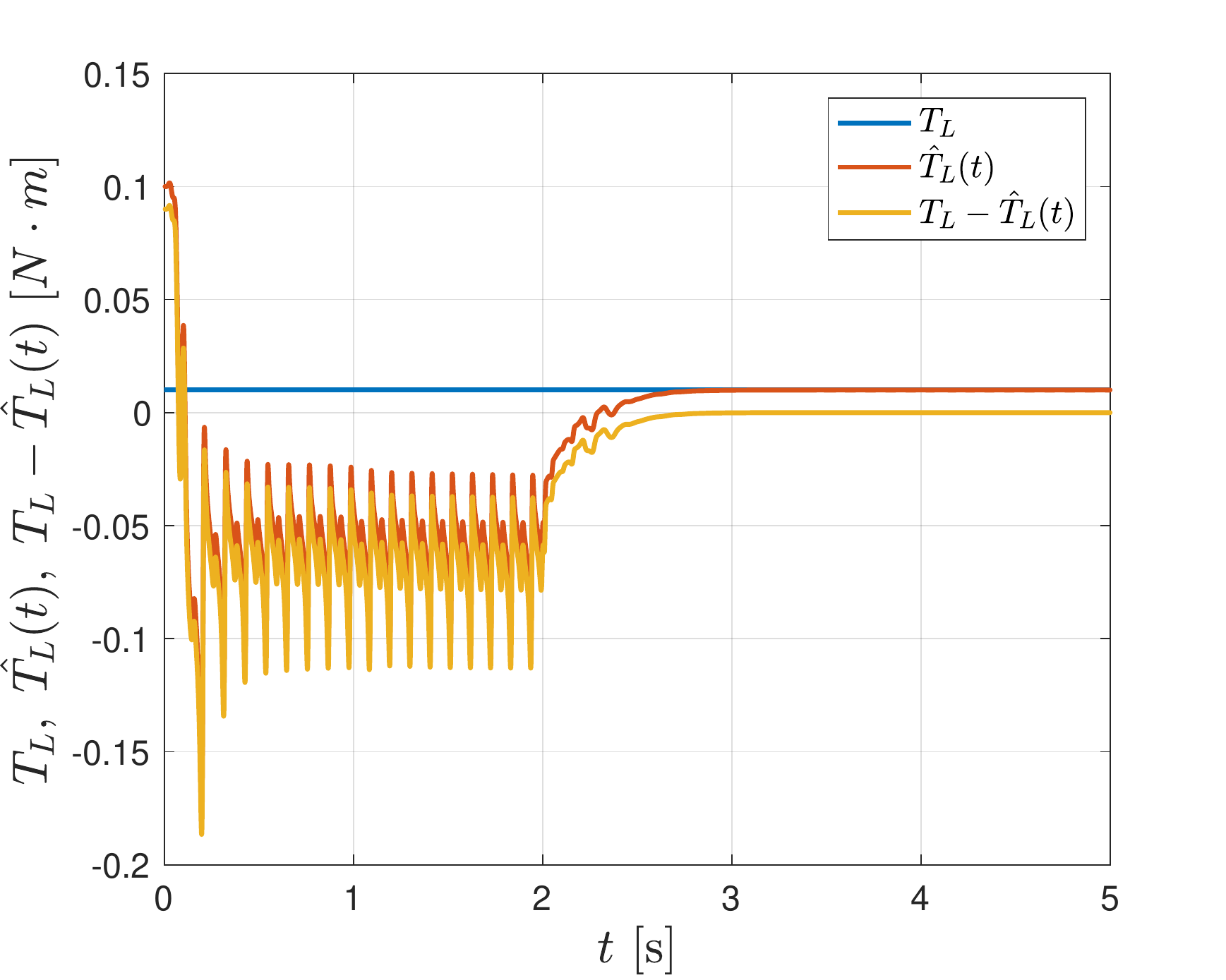}}	
	\vspace{-2mm}
	\caption{\label{fig:observer2} Transitions for external load torque and speed estimates }
\end{figure*}

\section{Conclusions}
\label{sec9}
The paper is devoted to the problem of adaptive observer design for IM. It is assumed that only stator voltages and currents are available for measurement and that all parameters of the IM except rotor resistance and load torque are known.  The solution to this problem is obtained applying the recent results on DREMBAO reported in \cite{PYRetalscl19}. The algorithms convergence relies on the verification of an excitation conditions---Assumptions \ref{ass1} and \ref{ass2}---which has been shown in \cite{ORTetaltac20} to be strictly weaker than PE of the regressor. 

Current research is underway to sharpen the aforementioned excitation conditions, with the hope of providing admissible verifiable operation modes of the IM. Another, very challenging, topic of interest is the application of the adaptive observers in closed-loop operation. Also, some preliminary results on adaptive observers with finite convergence time, with a very weak {\em interval excitation} assumption has been designed in \cite{OVCetal}.  Finally, we are currently working on the practical implementation of the proposed algorithms in an {\em experimental} benchmark.  We hope to be able to report these new results in the near future.


%
\appendix 
\section*{Proof of Lemma \ref{lem1}}
\label{app1}
The first step to prove the lemma is the key observation made in \cite[Subsection 1.2]{MARTOMVERbook} that the derivatives of the flux and the current are related via
\begin{align}
\beta \dot \lambda &=-{R_s} i + v-{\sigma L_s}{di\over dt},
\label{im2}
\end{align}
which follows directly from \eqref{im1_flux} and \eqref{im1_current}.

Let us introduce the following partial change of coordinates
\begin{align}
\lab{xi}
\xi=|\lambda|^2
\end{align}
and consider its derivative 
\begin{align}
\lab{xi_dot}
\dot\xi &= 2\lambda^\top\dot\lambda\nonumber\\
&=2\lambda^\top\left[  -\left({R_r\over L_r}I_2-n_p \calj\omega\right)\lambda+R_r\beta i\right]\nonumber\\
&= -2{R_r\over L_r}\xi+2 R_r\beta\lambda^\top i.
\end{align}

Now, consider a set of stable, linear time-invariant (LTI) filters ${\alpha_\ell \over p+\alpha_\ell }$, parameterized by the constant $\alpha_\ell >0$, where $p:={d\over dt}$. Applying this filters to the model \eqref{xi_dot} we get
\begin{align}
\lab{reg1}
{\alpha_\ell \over p+\alpha_\ell }[ \dot \xi] = -2 {R_r\over L_r} {\alpha_\ell \over p+\alpha_\ell }[\xi] + 2 R_r\beta{\alpha_\ell \over p+\alpha_\ell }[\lambda^\top i].
\end{align}

We will prove now that \eqref{reg1} may be represented as \eqref{lre}. Towards this end, we utilize the Swapping Lemma \cite[Lemma 3.6.5]{SASBOD} and some algebraic operations with the model equations \eqref{imdyn}.

Transform each term separately using the Swapping Lemma
\begin{align}
{\alpha_\ell\over p+\alpha_\ell} [\dot \xi] &=2{\alpha_\ell\over p+\alpha_\ell} [\lambda^\top\dot\lambda]
\nonumber\\
&=2\lambda^\top{\alpha_\ell\over p+\alpha_\ell} [\dot\lambda]-{2\over p+\alpha_\ell}
\left[\dot\lambda^\top{\alpha_\ell\over p+\alpha_\ell} \dot\lambda\right].
\end{align}
Taking $\dot\lambda$ from \eqref{im2} we get
\begin{align}
\lab{term1}
{\alpha_\ell\over p+\alpha_\ell} [\dot \xi]  &=2\lambda^\top{\alpha_\ell\over p+\alpha_\ell} \left[ {-R_s\over\beta}i+\frac{v}{\beta}-\frac{\sigma L_s}{\beta}{\dot{\hat{i}}}\right]
\nonumber\\&\quad
-{2\over p+\alpha_\ell}
\left[\left({-R_s\over\beta}i+\frac{v}{\beta}-\frac{\sigma L_s}{\beta}{\dot{\hat{i}}}\right)^\top \right.
\nonumber\\&\qquad
\times {\alpha_\ell\over p+\alpha_\ell} \left.\left[{-R_s\over\beta}i+\frac{v}{\beta}-\frac{\sigma L_s}{\beta}{\dot{\hat{i}}_{f}}\right]\right].
\end{align}

{
To simplify the notation, define the measurable filtered signals
\begin{align}
	 f_i &:= {\alpha_\ell \over p+\alpha_\ell }[i], \qquad
	\dot f_i:= {\alpha_\ell \,  p\over p+\alpha_\ell }[i], \nonumber\\
	f_v &:= {\alpha_\ell \over p+\alpha_\ell }[v], \qquad
	\dot f_v := {\alpha_\ell \, p\over p+\alpha_\ell }[v]. \nonumber
\end{align}
}

Using the definitions above rewrite \eqref{term1} as
\begin{align}
\lab{term12}
{\alpha_\ell \over p+\alpha_\ell } [ \dot \xi ]  &=-2{R_s\over\beta}\lambda^\top f_i+\frac{2}{\beta}\lambda^\top f_v-\frac{2 \sigma L_s}{\beta}\lambda^\top\dot{f_i}\nonumber\\
-{2\over \beta^2}{1\over p+\alpha_\ell }
&
\left[\left(-R_s i+v\right)^\top
\left(-R_s f_i+f_v-\sigma L_s\dot{f_i}\right)\right]
\nonumber\\
&
+{2 \sigma L_s \over \alpha_\ell  \beta^2}{\alpha_\ell \over p+\alpha_\ell }
\left[{di\over dt}^\top\left(-R_s f_i+f_v-\sigma L_s \dot{f_i}\right)\right].
\end{align}

Let us consider first the last right-hand term of \eqref{term12}, that is,
\begin{align*}
&{\alpha_\ell \over p+\alpha_\ell }
\left[{di\over dt}^\top\left(-R_s f_i+f_v-\sigma L_s\dot{f_i}\right)\right]=\nonumber\\
&=\left(-R_s f_i+f_v-\sigma L_s\dot{f_i}\right)^\top\dot{f_i}
\nonumber \\
& \quad -{1\over p+\alpha_\ell }\left[\left(-R_s\dot{f_i}+\dot f_v-\sigma L_s\ddot f_i\right)^\top\dot{f_i}\right]
\nonumber \\
&=\left(-R_s f_i+f_v-\sigma L_s\dot{f_i}\right)^\top\dot{f_i}
\nonumber \\
& \quad -{1\over p+\alpha_\ell }\left[\left(-R_s\dot{f_i}+\dot{ f_v}\right)^\top\dot{f_i}\right]+\sigma L_s{1\over p+\alpha_\ell }\left[\ddot f_i\,^\top\dot{f_i}\right]
\nonumber \\
&=\left(-R_s f_i+f_v-\sigma L_s\dot{f_i}\right)^\top\dot{f_i}
\nonumber \\
& \quad -{1\over p+\alpha_\ell }\left[\left(-R_s\dot{f_i}+\dot f_v\right)^\top\dot{f_i}\right]+\hal\sigma L_s {p\over p+\alpha_\ell }\left[\dot{f_i}^\top\dot{f_i}\right]
\nonumber \\
&=R_s\left(- f_i^\top\dot{f_i}+{1\over p+\alpha_\ell }\left[|\dot{f_i}|^2\right]\right)
\nonumber \\
& \quad +\sigma L_s \left(-\dot{f_i}^\top\dot{f_i} +\hal{p\over p+\alpha_\ell }\left[|\dot{f_i}|^2\right]\right)
\nonumber \\
& \quad +
f_v^\top\dot{f_i}-{1\over p+\alpha_\ell }\left[\dot f_v^\top\dot{f_i}\right].
\end{align*}

Replacing the latter identity in  \eqref{term12}, we see that ${\alpha_\ell \over p+\alpha_\ell } [\dot \xi] $ may be written as
\begin{align}
\lab{term13}
{\alpha_\ell \over p+\alpha_\ell } [\dot \xi] & = \lambda^\top \rho_1+\rho_2,
\end{align}
where we defined the  measurable signals
\begin{align*}
	\rho_1&:=\frac{2}{\beta}\left(-R_s  f_i-\sigma L_s\dot{f_i}+f_v\right),
	\\	\rho_2&:=-\frac{2}{\beta^2}{1\over p+\alpha_\ell }[v^\top f_v]
\nonumber \\
& \quad +R_s \mu_1+\sigma  \mu_2+R_s^2 \mu_3+R_s\sigma \mu_4+\sigma^2 \mu_5,
\end{align*}
and
\begin{align*}
	\mu_1&:={2\over\beta^2}{1\over p+\alpha_\ell }\left[i^\top f_v +v^\top f_i\right],
	\\
	\mu_2&:={2 L_s \over\alpha_\ell  \beta^2}f_v^\top\dot{f_i}+\frac{2 L_s}{\beta^2}{1 \over p+\alpha_\ell }\left[v^\top\dot{f_i}-{1\over \alpha_\ell }\left(\dot f_v^\top\dot{f_i}\right)\right],
	\\
	\mu_3&:=-{2\over\beta^2}{1\over p+\alpha_\ell }[i^\top f_i],
	\\
	\mu_4&:=-{2 L_s \over\beta^2}{1\over p +\alpha_\ell }\left[i^\top\dot{f_i}\right]\nonumber\\
	&\quad+{2 L_s \over\alpha_\ell \beta^2}\left(- f_i^\top\dot{f_i}+{1\over p+\alpha_\ell }\left[|\dot{f_i}|^2\right]\right),
	\\
	\mu_5&:={2 L_s^2 \over\alpha_\ell  \beta^2}\left(-|\dot{f_i}|^2 +\hal{p\over p+\alpha_\ell }\left[|\dot{f_i}|^2\right]\right).
\end{align*}

Combining  \eqref{reg1} and \eqref{term13} we get the identity
\begequ
\lab{reg11}
 \lambda^\top \rho_1+\rho_2 = -2 {R_r\over L_r} {\alpha_\ell \over p+\alpha_\ell }[\xi] + 2 R_r\beta{\alpha_\ell \over p+\alpha_\ell }[\lambda^\top i].
\endequ
We proceed now to analyze the the two right-hand terms of \eqref{reg11}. For the first term we get
\begin{align}
\lab{term2}
{\alpha_\ell \over p+\alpha_\ell }[\xi]
&=\xi-{1\over p+\alpha_\ell }[\dot \xi] + \epsilon_1(t)
\nonumber\\
&=\xi-{1\over \alpha_\ell }\left(\lambda^\top \rho_1 +\rho_2\right) + \epsilon_1(t),
\end{align}
where $\epsilon_1(t):=\frac{\alpha_\ell }{p + \alpha_\ell } [1(t)]$ is exponentially decaying term.

Finally, consider the second term of \eqref{reg11}
\begin{align}
\lab{term3}
{\alpha_\ell \over p+\alpha_\ell }[\lambda^\top i]&=\lambda^\top f_i-{1\over\beta}{1\over p+\alpha_\ell }\left[\left(-R_s i+v-\sigma L_s{di\over dt}\right)^\top f_i\right]\nonumber\\
&=\lambda^\top f_i+R_s{1\over\beta}{1\over p+\alpha_\ell }[i^\top f_i]
\nonumber \\
& \quad -\frac{1}{\beta}{1\over p+\alpha_\ell }[v^\top f_i]+
\frac{L_s \sigma}{\beta}{1\over p+\alpha_\ell }\left[{di^\top\over dt} f_i\right]\nonumber\\
&=\lambda^\top f_i+\rho_3
\end{align}
where  we defined the  measurable signal
\begin{align*}
\rho_3&:=R_s{1\over\beta}{1\over p+\alpha_\ell }[i^\top f_i]
 -\frac{1}{\beta}{1\over p+\alpha_\ell }[v^\top f_i]
\nonumber \\
& \quad +\frac{L_s \sigma}{\alpha_\ell  \beta}\left( f_i^\top\dot{f_i}
-{1\over p+\alpha_\ell }\left[|\dot{f_i}|^2\right]\right).
\end{align*}
The proof is completed by replacing \eqref{term2} and \eqref{term3} in \eqref{reg11}, grouping terms, and defining\footnote{We underscore the fact that, for each constant $\alpha_\ell$, we generate {\em different} signals $\rho_1,\rho_2,\rho_3$, $f_i$ and $f_v$, but this dependence is omitted to simplify the notation.}
\begin{align}
z(t, \alpha_\ell ) &:=\rho_2\\
\phi_{e}(t, \alpha_\ell) &:=
\begmat{
{2\over\alpha_\ell L}\rho_2+2\beta\rho_3 \\
-\rho_1\\
{2\over\alpha_\ell L}\rho_1+2\beta f_i\\
-{2\over L}} .
\end{align}
%
\section*{Proof of Lemma \ref{lem2}}
\label{app2}
%
Following the construction of the dynamically extended regressor we consider the set of LTI filters $\frac{\alpha_\ell }{p + \alpha_\ell }$ with different $\alpha_\ell >0$, $\ell = 1,\dots,6$, to generate different filtered signals $z(t, \alpha_\ell )$ and $\phi_{e}(t, \alpha_\ell)$. Piling this signals up and using \eqref{lre} we  obtain a matrix equation
\begin{align}
\lab{psiphithe}
	\Psi(t) = \Phi(t) \Theta(t),
\end{align}
where
\begin{align*}
	\Psi(t) & := \col \{z(t, \alpha_1 ), \dots, z(t, \alpha_6 )\} \in \mathbb{R}^6,
	\\
	\Phi(t) & :=  \col \{\phi_{e}^\top(t, \alpha_1 ), \dots, \phi_{e}^\top(t, \alpha_6 )\} \in \mathbb{R}^6,
	\\
	\Theta(t) & := \begmat{R_r \\ \lambda(t) \\ R_r\lambda(t) \\ R_r|\lambda(t)|^2} \in \mathbb{R}^6.
\end{align*}
Applying the next step of the DREM procedure we multiply \eqref{psiphithe} by the adjugate of the matrix $\Phi(t)$, denoted $\adj \{\Phi(t) \}$ to get \eqref{scalre} with the definitions
\begin{align}
\lab{zetae}
	\zeta_e(t) & := \adj \{\Phi(t) \} \Psi(t),
	\\
	\lab{deltae}
	 \Delta_e(t) & := \det\{\Phi(t) \}.
\end{align} 
\end{document}